\documentclass[a4paper]{jpsj2}
\makeatletter
\@ifpackageloaded{graphicx}{}{\usepackage{graphicx}}
\makeatother

\usepackage{amsthm,bm}
\usepackage{mathtools}
\makeatletter
\providecommand{\ext@table}{lot}
\providecommand{\ext@figure}{lof}
\makeatother
\usepackage[hidelinks]{hyperref}
\hypersetup{
    pdftitle={Finite-Mask Gaussian-Process Reconstruction on a Periodic Sensor Ring: Mask-Geometry Dependence, Fourier-Mode Coupling, and Posterior Trace},
    pdfauthor={Jun Tsuzurugi},
    pdfsubject={Author manuscript corresponding to J. Phys. Soc. Jpn. 95, 104002 (2026)}
}

\theoremstyle{definition}
\newtheorem{definition}{Definition}
\theoremstyle{plain}
\newtheorem{theorem}{Theorem}
\newtheorem{proposition}{Proposition}
\newtheorem{corollary}{Corollary}
\theoremstyle{definition}
\newtheorem{remark}{Remark}
\theoremstyle{plain}

\newcommand{\TL}{\mathbb{T}_{L}}
\newcommand{\E}{\mathbb{E}}
\providecommand{\Tr}{\operatorname{Tr}}
\newcommand{\diag}{\operatorname{diag}}

\def\runauthor{J. Tsuzurugi}

\title{

Finite-Mask Gaussian-Process Reconstruction on a\\
Periodic Sensor Ring:\\
Mask-Geometry Dependence, Fourier-Mode Coupling, and Posterior Trace

}

\author{\href{https://orcid.org/0009-0004-6240-4550}{Jun Tsuzurugi}\thanks{\href{mailto:juntuzu@ous.ac.jp}{juntuzu@ous.ac.jp}}}

\inst{Department of Information Science and Engineering,\\
Okayama University of Science, Okayama 700-0005, Japan}

\recdate{June 22, 2026; revised July 18, 2026; accepted August 10, 2026;
published online September 9, 2026}

\abst{

We analyze how an arbitrary observation mask on a finite periodic ring of
equally spaced sensors affects the posterior covariance, Fourier-mode coupling,
and normalized posterior trace in Gaussian-process reconstruction.  Under a
rotationally stationary prior and homogeneous independent measurement noise,
complete observation gives independent scalar posterior formulas for the
Fourier modes.  For an arbitrary mask, however, the matrix
\(Q_M=F D_MF^\ast\) is generally non-diagonal; its off-diagonal entries are
finite Fourier components of the realized mask and couple modes in the
posterior precision.  Consequently, masks with the same unavailable-channel
fraction can have different normalized posterior traces because their
geometries differ.  A dimensionless 64-channel synthetic benchmark illustrates
this finite-matrix effect.  A circumferential array of equally spaced
wall-mounted microphones at a fixed axial station of a circular fan or
compressor duct provides one concrete mechanical-engineering interpretation:
failed, saturated, corrupted, or dropped-out channels form the observation
mask.  The analysis is a finite-dimensional reference calculation under the
stated rotational-stationarity and common-noise assumptions, not a performance
claim for a nonuniform or unequally instrumented operating duct.

}

\begin{document}

\maketitle
\enlargethispage{2pt}

\begin{center}
\small
Author manuscript of J. Phys. Soc. Jpn. \textbf{95}, 104002 (2026).\\
\url{https://doi.org/10.7566/JPSJ.95.104002}
\end{center}

\section{Introduction}

Partially observed finite data are common in Bayesian reconstruction.  Sensor
records may contain dropouts, masked signal or lattice data may be available
only at selected sites, and spatial measurements may be restricted by accessible
sampling locations.  In such situations, complete observation of every lattice
site is an idealization.

Gaussian-process reconstruction is useful because it gives both a posterior
mean and a posterior covariance
\cite{RasmussenWilliams,Bishop,MacKay}.  The posterior covariance measures the
uncertainty that remains after conditioning on the observation mask, and it
connects reconstruction reliability with sampling design and missing-data
assessment.  The practical question is not only how to reconstruct unobserved
values, but how reliable that reconstruction is for the actual measurement
pattern.

The observed-site count gives only a partial answer: regularly distributed and
clustered observations can leave different posterior uncertainties even at the
same missing rate.  Thus the missing rate is a coarse summary, and the posterior
formulation must retain the geometry that enters \(\Sigma_M\) together with the
count.

A finite periodic lattice gives a minimal controlled setting in which this
geometry dependence can be separated from boundary effects.  With a
translation-invariant Gaussian prior, complete observation is exactly diagonal in
the Fourier basis and gives independent scalar Fourier-mode posterior formulas.
Once an arbitrary observation mask is inserted, the likelihood precision is no
longer diagonal in that basis; the mask couples Fourier modes and changes the
posterior uncertainty even though the prior remains Fourier diagonal.

This paper uses the standard finite-dimensional linear-Gaussian posterior
identity as a controlled starting point rather than as a claimed new conditioning
theorem.  The contribution is to embed that identity in the periodic-lattice
Fourier representation with an arbitrary observation mask and to isolate
\(D_M\), or equivalently \(Q_M=F D_MF^\ast\), as an explicit finite matrix.
Under complete observation, \(D_M=I_L\) and the posterior reduces to independent
scalar Fourier modes.  Under a sparse mask, the prior remains diagonal as
\(K=F^\ast\Lambda F\), but the likelihood contribution \(b^{-2}Q_M\) generally
couples those modes.  The matched error is then the normalized posterior trace
\(L^{-1}\operatorname{Tr}\Sigma_M\), replacing the independent scalar-mode error
sum of the complete-observation case.

The finite periodic model serves as a controlled benchmark for partially
observed lattice or signal data.  It keeps the Gaussian-process prior and the
observation geometry separated while remaining exactly finite-dimensional.
Connections with spatial-statistical kriging, Wiener filtering, and masked
finite-signal analysis are made explicit below.

\subsection{Engineering motivation and circumferential microphone-array interpretation}
\label{sec:duct-interpretation}

Circumferential arrays of wall-mounted microphones are used to resolve
circumferential, or spinning, acoustic modes in circular fan and compressor
ducts.  Joppa applied a spatial Fourier transform to an equally spaced array to
measure spinning modes in a turbofan inlet \cite{Joppa1987}; Wang et al. used an
equidistant circular array for in-duct circumferential-mode measurements of an
axial fan/compressor \cite{WangEtAl2014}; and Wang et al. reconstructed a full
cross-spectral matrix from incomplete circumferential measurements before
identifying broadband acoustic modes \cite{WangEtAl2023}.  Such measurements
provide modal information for characterizing duct sound fields and
turbomachinery-noise propagation, and hence are important for noise-source
characterization and noise assessment in mechanical engineering and
aeroacoustics.

Microphone failure, saturation, recording corruption, or communication dropout
produces incomplete circumferential samples.  The resulting task is therefore
not only to reconstruct pressure values at unavailable channels, but also to
quantify the uncertainty conditional on the realized arrangement of those
channels.  This is the engineering connection to the arbitrary finite mask
analyzed below.

A concrete engineering interpretation of the finite periodic sensor ring is a
circumferential microphone array installed at one fixed axial station of an
ideal circular fan or compressor duct.  Let \(L\) nominally identical
wall-mounted microphones be equally spaced at angular positions
\begin{equation}
    \theta_j=\frac{2\pi j}{L},
    \qquad j=0,1,\ldots,L-1.
\end{equation}
The discrete torus \(\mathbb{T}_L\) then represents the angular ordering of the
microphones, including the periodic adjacency between channels \(L-1\) and
zero.  For a simultaneous measurement, let
\(S_j=S(\theta_j)\) denote the mean-removed real-valued acoustic-pressure
fluctuation at the \(j\)-th microphone.  An available channel records
\begin{equation}
    t_j=S_j+\xi_j,
\end{equation}
where the present model assumes independent, identically distributed
measurement noise with variance \(b^2\).

The translation-invariant covariance used below follows from explicit
second-order assumptions.  Suppose that the duct and the nominal sensor layout
are rotationally uniform and that, after mean removal, the probability law of
the pressure fluctuation is invariant under a common circumferential rotation.
Then its covariance depends only on angular separation:
\begin{equation}
\begin{split}
    K_{jj'}
    &=\operatorname{Cov}\!\left[S(\theta_j),S(\theta_{j'})\right] \\
    &=C\!\left((\theta_j-\theta_{j'})\bmod 2\pi\right) \\
    &=K_{(j-j')\bmod L}.
\end{split}
\label{eq:duct-circulant-covariance}
\end{equation}
Thus \(K\) is a circulant matrix.  This assumption does not mean that each
realized pressure waveform is rotationally symmetric.  It means that the
mean-removed probability distribution, or at least its second-order
covariance, is invariant under a common angular shift.

The same conclusion has a circumferential-mode interpretation.  Write the
mean-removed pressure field as
\begin{equation}
    S(\theta)=\sum_{n\in\mathbb{Z}}a_n e^{in\theta},
    \qquad
    \mathbb{E}[a_n]=0,
\end{equation}
with \(a_{-n}=a_n^\ast\) for a real field.  If distinct circumferential-mode
coefficients are uncorrelated and
\begin{equation}
    \mathbb{E}[a_n a_{n'}^\ast]=P_n\delta_{nn'},
    \qquad P_n\ge0,
\end{equation}
then
\begin{equation}
    K_{jj'}
    =\sum_{n\in\mathbb{Z}}P_n
      \exp\!\left\{in(\theta_j-\theta_{j'})\right\}.
\label{eq:duct-mode-covariance}
\end{equation}
For the \(L\) equally spaced samples, continuous circumferential orders that
are congruent modulo \(L\) are aliased into the same discrete mode.  With the
unitary Fourier matrix defined in Sect.~2,
\begin{equation}
    F K F^\ast
    =\diag(\lambda_0,\ldots,\lambda_{L-1}),
    \qquad
    \lambda_q=L\sum_{r\in\mathbb{Z}}P_{q+rL}\ge0.
\label{eq:duct-fourier-diagonalization}
\end{equation}
Fourier diagonalization is therefore not merely a computational assumption in
this idealized measurement system; it is the discrete circumferential acoustic
mode decomposition of a rotationally stationary covariance.

The binary mask has the direct channel interpretation
\(M_j=1\) when microphone \(j\) provides a usable record and \(M_j=0\) when
that channel is unavailable because of microphone failure, saturation, a
corrupted recording, or a communication dropout.  The row-selection matrix
\(H\) retains only the usable channels, and
\(D_M=H^T H=\operatorname{diag}(M_0,\ldots,M_{L-1})\) records their
circumferential acquisition geometry.  Unlike \(K\), an arbitrary \(D_M\) is
not rotationally invariant.  Its Fourier representation has entries
\begin{equation}
    (Q_M)_{qq'}
    =(F D_MF^\ast)_{qq'}
    =\frac{1}{L}\sum_{j=0}^{L-1}M_j
      e^{-i(\kappa_q-\kappa_{q'})j},
\label{eq:duct-mask-mode-mixing}
\end{equation}
which are generally nonzero for \(q\ne q'\).  Hence the off-diagonal elements
of \(Q_M\) mix circumferential acoustic modes in the posterior precision
\(\Lambda^{-1}+b^{-2}Q_M\).  This is the direct physical interpretation of the
mask-induced Fourier-mode coupling studied in this paper.

Table~\ref{tab:duct-notation} summarizes the correspondence between the
mathematical notation and the microphone-array problem.

\begin{table}[h]

\centering
\caption{
Correspondence between the finite periodic model and a circumferential
microphone array at one axial station.
}
\label{tab:duct-notation}
\begin{tabular}{p{0.27\linewidth} p{0.63\linewidth}}
\hline
Mathematical quantity & Circumferential-array interpretation \\
\hline
\(L\), \(\mathbb{T}_L\) & Number of nominal channels and their cyclic angular ordering. \\
\(j\), \(\theta_j=2\pi j/L\) & Microphone index and its circumferential angle. \\
\(S_j=S(\theta_j)\) & True mean-removed acoustic-pressure fluctuation at channel \(j\). \\
\(t_j\), \(j\in\Omega\) & Noisy measurement returned by an available channel. \\
\(b^2 I_m\) & Homogeneous independent measurement-noise covariance for the \(m\) available channels. \\
\(K\) & Circumferential pressure covariance assumed for the nominal rotationally stationary field. \\
\(F\), \(\widetilde S_n\), \(\lambda_n\) & Circumferential discrete-mode transform, modal pressure coefficient, and prior modal variance. \\
\(M_j\), \(\Omega\), \(m\) & Channel-availability indicator, set of available channels, and their number. \\
\(H\), \(D_M=H^T H\) & Usable-channel selector and full-ring diagonal availability matrix. \\
\(Q_M=F D_MF^\ast\) & Missing-channel-induced mixing matrix between circumferential modes. \\
\(\Sigma_M\) & Posterior error covariance for the reconstructed full-ring pressure vector. \\
\(G_{\max}\), \(2\pi G_{\max}/L\) & Largest cyclic gap in channel-index units and the corresponding largest angular gap between available microphones. \\
\(p(M)=1-m/L\) & Realized fraction of unavailable channels. \\
\(E(M)=L^{-1}\operatorname{Tr}\Sigma_M\) & Posterior variance averaged over the entire circumference. \\
\hline
\end{tabular}

\end{table}

The posterior mean reconstructs the full pressure vector, including the
unavailable channel positions, whereas \(\Sigma_M\) and \(E(M)\) quantify the
uncertainty that remains for the realized microphone configuration.

The present study does not analyze engine or compressor data.  It is an
idealized finite-dimensional reference model for this measurement configuration
and does not claim reconstruction performance or model fit for an operating
machine.  Strong support-structure effects, circumferentially nonuniform swirl,
nonuniform wall conditions, or other azimuthal inhomogeneities can make \(K\)
non-circulant, while sensor-specific noise can invalidate the common-noise model
\(b^2I_m\).  Such systems require a nonstationary covariance or a
channel-dependent likelihood and lie outside the present scope.  The narrower
purpose here is to determine exactly how a realized pattern of unavailable
circumferential channels enters Gaussian reconstruction and its posterior
uncertainty in the ideal rotationally stationary reference case.

\subsection{Related work and distinction from neighboring methods}
\label{sec:related-work}

Circumferential microphone measurements and spatial Fourier analysis of duct
acoustic modes provide the experimental context for the sensor-ring
interpretation.  Joppa analyzed equally spaced, wall-mounted circumferential
microphone data by a spatial Fourier transform to measure spinning modes in a
turbofan inlet \cite{Joppa1987}.  Wang et al. used an equidistant circular
microphone array for in-duct circumferential-mode measurements of an axial
fan/compressor \cite{WangEtAl2014}.  For nonsynchronous measurements with a
wall-mounted circular array, Wang et al. reconstructed a full cross-spectral
matrix from incomplete circumferential measurements and then identified
broadband acoustic modes \cite{WangEtAl2023}.  These works primarily address
acoustic-mode measurement or identification, and the last also addresses
cross-spectral matrix completion.  The present paper instead conditions a
Gaussian field on an arbitrary fixed binary pattern of unavailable channels,
evaluates the resulting posterior covariance and its normalized trace exactly
in finite dimensions, and isolates \(Q_M=F D_MF^\ast\) as the mechanism by
which two masks with the same unavailable fraction can yield different
uncertainties.  It neither proposes a duct-mode identification procedure nor a
cross-spectral matrix-completion algorithm.

In the terminology of spatial statistics, the same Gaussian conditioning formula
underlies kriging for partially observed spatial fields
\cite{Cressie1993,Stein1999}.  In the terminology of Gaussian-process
regression, the same finite Gaussian conditioning formulas are commonly used
\cite{RasmussenWilliams,Bishop,MacKay}.  The present paper uses these formulas as
a finite periodic-lattice benchmark with an arbitrary realized mask, and it
exposes the resulting Fourier-mode coupling.  The distinction is therefore not a new
conditioning identity, but the exact finite representation of the mask term and
its posterior-trace consequence in the periodic Fourier basis.

In signal-processing language, the complete-observation periodic case is the
finite Fourier-domain analogue of a scalar Wiener filter
\cite{Wiener1949}.  When every site is observed and the prior covariance is
translation invariant, each Fourier mode is shrunk independently by a scalar
factor.  With an arbitrary mask, however, the observation operator breaks
translation symmetry.  The likelihood contribution in Fourier space is then
\(b^{-2}Q_M\), with
\begin{equation}
    Q_M=F D_MF^\ast,
\end{equation}
and the posterior precision becomes \(\Lambda^{-1}+b^{-2}Q_M\).  Since
\(Q_M\) is generally non-diagonal, the posterior is a coupled finite-mode
problem rather than a collection of independent scalar Wiener shrinkages.

This is also the finite-lattice version of the familiar Fourier effect of
masking a signal.  Multiplication by a real-space mask corresponds in Fourier
space to convolution by the Fourier components of that mask, often described as
spectral mixing or leakage in finite-signal analysis
\cite{OppenheimSchafer2010}.  In the present Bayesian formulation this effect
appears directly inside the posterior precision through \(Q_M=F D_MF^\ast\), whose
off-diagonal entries are the finite Fourier components of the realized mask.

The formulation is also distinct from compressed sensing or sparse-recovery
theory \cite{CandesWakin2008}.  The latent signal is modeled by a Gaussian prior,
and the evaluated quantities are the posterior covariance and the normalized
posterior trace.  The aim is to quantify Bayes uncertainty for each finite
realized mask.

\paragraph{Scope.}
The scope of the paper is deliberately finite and mask-conditioned.  We do not
derive an asymptotic sampling theorem or an off-lattice interpolation theorem.
Instead, we focus on the exactly computable posterior covariance for each
realized finite observation mask and on the resulting dependence of posterior
uncertainty on mask geometry.  Random masks and random inter-observation
intervals are used as generation and reporting layers for realized finite masks.

\paragraph{Notation bridge to the complete-observation formulation.}
In the previous complete-observation formulation~\cite{Tsuzurugi2025},
the notation \(N,A,B,C=A+B\) was used.  Table~\ref{tab:notation-bridge}
summarizes how these symbols are carried into the present sparse-observation
formulation.  The previous formulation is preserved and recovered exactly as the boundary case
\(H=I_L\), \(D_M=I_L\), and \(m=L\).
Below, after this bridge, that complete-observation work is called the previous
formulation.

\begin{table}[h]

\centering
\caption{
Notation bridge between the previous complete-observation formulation and the
present sparse-observation formulation.  The complete-observation covariance
\(K+b^2I_L\) is recovered in the boundary case \(H=I_L\),
\(D_M=H^TH=I_L\), and \(m=L\).
}
\label{tab:notation-bridge}
\begin{tabular}{p{0.34\linewidth} p{0.56\linewidth}}
\hline
Previous complete-observation formulation & Current sparse-observation formulation \\
\hline
Lattice size \(N\) & Full latent lattice size \(L\); the number of observed sites is \(m=|\Omega|\le L\). \\
Signal covariance \(A\) & Full-lattice prior covariance \(K\). \\
Noise covariance \(B=b^2I_N\) & Observation-space noise covariance \(b^2I_m\). \\
Covariance sum \(C=A+B\) & Observation covariance \(HKH^T+b^2I_m\). \\
Complete observation & Boundary case \(H=I_L\), \(D_M=H^TH=I_L\), and \(m=L\), giving \(HKH^T+b^2I_m=K+b^2I_L\). \\
\hline
\end{tabular}

\end{table}

The contribution is threefold.  First, the paper fixes a finite-mask posterior
notation on the same finite periodic lattice as the complete-observation theory,
without claiming a new Gaussian-conditioning identity.  Second, it identifies the
mask term \(Q_M=F D_MF^\ast\) as the object that is generally non-diagonal in the
Fourier basis; its off-diagonal entries are finite Fourier components of the
realized mask and convert the scalar Fourier-mode posterior into a coupled
finite-mode posterior.  Third, it uses the posterior trace to quantify the
dependence of reconstruction uncertainty on realized mask geometry, thereby
making explicit why the missing rate should be supplemented by mask-geometry
information.

\section{One-dimensional Periodic Lattice and Admissible Covariance}

\begin{definition}[One-dimensional finite torus]
Let
\begin{equation}
    \TL=\{0,1,\ldots,L-1\}
\end{equation}
be the one-dimensional discrete torus.  The Fourier index is denoted by
\(n\in\TL\), and
\begin{equation}
    \kappa_n=\frac{2\pi n}{L}.
\end{equation}
\end{definition}

For a vector \(S=(S_0,\ldots,S_{L-1})^T\), we use the unitary discrete Fourier
transform
\begin{equation}
    \widetilde{S}_n
    =
    \frac{1}{\sqrt{L}}
    \sum_{j=0}^{L-1} S_j e^{-i\kappa_n j},
    \qquad
    S_j
    =
    \frac{1}{\sqrt{L}}
    \sum_{n=0}^{L-1} \widetilde{S}_n e^{i\kappa_n j}.
\end{equation}

\begin{definition}[Admissible periodic covariance kernel]
A translation-invariant periodic kernel \(K_r\), \(r\in\TL\), is admissible
if its finite Fourier eigenvalues
\begin{equation}
    \lambda_n
    =
    \sum_{r=0}^{L-1} K_r e^{-i\kappa_n r}
\end{equation}
satisfy
\begin{equation}
    \lambda_n\ge 0
    \qquad
    (n\in\TL).
\end{equation}
For real fields we also assume \(K_{-r}=K_r\).
\end{definition}

The covariance matrix is
\begin{equation}
    K_{ij}=K_{i-j},
\end{equation}
where all indices are understood on the finite torus.  It is diagonalized as
\begin{equation}
    K=F^{\ast}\Lambda F,
    \qquad
    \Lambda=\diag(\lambda_0,\ldots,\lambda_{L-1}),
\end{equation}
where \(F\) is the unitary Fourier matrix.

\subsection{Minimum-image Gaussian example}

For comparison with the previous periodic formulations, one may use the
minimum-image Gaussian
\begin{equation}
    K_r^{\mathrm{min}}
    =
    a^2
    \exp\left[
        -\frac{\widetilde{\delta}(r)^2}{2\ell^2}
    \right],
    \qquad
    \widetilde{\delta}(r)=\min\{r,L-r\}.
\label{eq:minimum-image-gaussian}
\end{equation}
Its finite eigenvalues are
\begin{equation}
    \lambda_n^{\mathrm{min}}
    =
    a^2
    \sum_{r=0}^{L-1}
    \exp\left[
        -\frac{\widetilde{\delta}(r)^2}{2\ell^2}
    \right]
    e^{-i\kappa_n r}.
\end{equation}
This kernel is used in parameter ranges satisfying
\begin{equation}
    \min_n \lambda_n^{\mathrm{min}}\ge 0.
\label{eq:minimum-image-admissibility}
\end{equation}

\section{Sparse-observation Model}

Let \(S_j\) be the latent zero-mean Gaussian field on the full lattice:
\begin{equation}
    S\sim \mathcal{N}(0,K).
\end{equation}
The sparse-observation formulation separates three layers.  The latent field
\(S\) is always an \(L\)-dimensional vector on the full lattice.  The observed
data \(t\) are an \(m\)-dimensional vector in observation space.  The diagonal
matrix \(D_M\) is an \(L\times L\) full-lattice bookkeeping device for the mask,
whereas \(H\) is the actual \(m\times L\) observation operator that maps the
latent lattice field to observation space.

Sparse observations are represented by a mask vector
\begin{equation}
    M_j\in\{0,1\},
    \qquad
    j\in\TL.
\end{equation}
Here and below, \(M=(M_0,\ldots,M_{L-1})\) is the binary mask vector, whereas
\(D_M\) is the diagonal matrix generated from that vector.  We do not use
\(M\) itself as a matrix.  Let
\begin{equation}
    \Omega=\{j\in\TL:M_j=1\},
    \qquad
    m=|\Omega|.
\end{equation}

\begin{definition}[Arbitrary finite observation mask]
A sparse-observation mask is a binary vector
\begin{equation}
    M=(M_0,\ldots,M_{L-1}),
    \qquad
    M_j\in\{0,1\}.
\end{equation}
The observed set is
\begin{equation}
    \Omega=\{j\in\TL:M_j=1\}.
\end{equation}
In the finite-dimensional formulation, no deterministic condition is imposed
on the positions of the observed sites, on the lengths of missing intervals, or
on the maximum gap.
\end{definition}

Define the observation operator \(H\in\mathbb{R}^{m\times L}\) by selecting
the rows of the identity matrix corresponding to \(\Omega\).  The observation
model is
\begin{equation}
    t=HS+\xi,
    \qquad
    \xi\sim\mathcal{N}(0,b^2 I_m).
\end{equation}
Equivalently, in full-lattice notation,
\begin{equation}
    t_j=S_j+\xi_j
    \qquad
    (j\in\Omega),
\end{equation}
where \(t_j\) denotes the component of the observation-space vector \(t\)
associated with the observed site \(j\in\Omega\).  Thus \(t\) is an
\(m\)-dimensional observation-space vector indexed by the observed set
\(\Omega\), or by any fixed ordering of that set.

The diagonal mask matrix on the full lattice is
\begin{equation}
    D_M=H^T H
    =
    \diag(M_0,\ldots,M_{L-1}).
\end{equation}

Figure~\ref{fig:mask-geometry} illustrates the finite mask geometry as two
64-channel circumferential microphone-array configurations.  In the derivation,
the realized mask enters the posterior formula through the observation operator
\(H\) and the full-lattice mask matrix \(D_M=H^T H\).

\section{Statistical-mechanical Posterior Formulation}

The prior density is
\begin{equation}
    P(S)
    =
    \frac{1}{Z_S}
    \exp\left[
        -\frac{1}{2}S^T K^{-1}S
    \right],
\end{equation}
assuming that \(K\) is strictly positive definite.  Appendix~\ref{app:zero-eigenvalue}
records the zero-eigenvalue convention used when this assumption is relaxed.

The sparse-observation likelihood is
\begin{equation}
    P(t|S)
    =
    \frac{1}{(2\pi b^2)^{m/2}}
    \exp\left[
        -\frac{1}{2b^2}\|t-HS\|^2
    \right].
\end{equation}
Therefore, up to a term independent of \(S\), the conditional Hamiltonian is
\begin{equation}
    H(S|t)
    =
    \frac{1}{2}S^T K^{-1}S
    +
    \frac{1}{2b^2}\|t-HS\|^2.
\end{equation}

The core object in this section is the exact finite-dimensional Gaussian
posterior conditioned on one realized observation mask; this is the finite-mask
continuation of the previous complete-observation posterior.  For a fixed
realized finite mask, all observation geometry enters the quadratic Hamiltonian
through the row-selection operator \(H\), or equivalently through the full-lattice
mask matrix \(D_M=H^TH\).

In this finite Gaussian model, the posterior precision matrix is the Hessian of
\(H(S|t)\) with respect to the latent field \(S\).  Thus the
statistical-mechanical calculation reduces to completing the square in a
quadratic Hamiltonian.  The important finite-mask change is not the conditioning
identity itself, but the replacement of the complete-observation likelihood
precision \(b^{-2}I_L\) by the mask precision \(b^{-2}D_M\).  In Fourier space
this term becomes \(b^{-2}F D_MF^\ast\), and this is the finite-matrix origin of
the Fourier-mode coupling appearing below.

The theorem is stated not as a new Gaussian-conditioning theorem, but as the
finite-mask notation used throughout the paper.  Its role is to display, in the
same statement, the latent-space posterior covariance, the observation-space
Woodbury form, and the Fourier-space precision containing the mask matrix
\(Q_M\).  The underlying conditioning step is the standard finite
linear-Gaussian identity \cite{Bishop,MacKay,Harville1997}; the notation below
isolates the mask-dependent Fourier mechanism and the trace quantity
\(L^{-1}\operatorname{Tr}\Sigma_M\).  The theorem is stated for
positive-definite \(K\), so that \(K^{-1}\) is an ordinary inverse.  The
positive-semidefinite convention is a technical support convention and is
collected in Appendix~\ref{app:zero-eigenvalue}.

\begin{theorem}[Finite-mask posterior and coupled Fourier representation]
Let \(S\sim\mathcal{N}(0,K)\) be a Gaussian field on \(\TL\), and let
\[
    t=HS+\xi,
    \qquad
    \xi\sim\mathcal{N}(0,b^2I_m),
\]
where \(H\) selects the observed sites specified by a fixed realized finite
mask \(M\).  If \(K\) is positive definite and \(b^2>0\), then the posterior
distribution \(P(S|t)\) is Gaussian with covariance
\begin{equation}
    \Sigma_M
    =
    \left(K^{-1}+b^{-2}D_M\right)^{-1}
\end{equation}
and mean
\begin{equation}
    \mu_M
    =
    \Sigma_M b^{-2}H^Tt.
\end{equation}
Equivalently, by the Woodbury identity, the same posterior has the
observation-space implementation form with an \(m\times m\) inverse,
\begin{equation}
    \mu_M
    =
    KH^T(HKH^T+b^2I_m)^{-1}t
\end{equation}
and
\begin{equation}
    \Sigma_M
    =
    K-KH^T(HKH^T+b^2I_m)^{-1}HK.
\end{equation}
Writing \(\widetilde{\Sigma}_M=F\Sigma_MF^\ast\), the Fourier-space posterior
precision is
\begin{equation}
    \widetilde{\Sigma}_M^{-1}
    =
    \Lambda^{-1}+b^{-2}Q_M,
    \qquad
    Q_M=F D_MF^\ast .
\end{equation}
Thus the finite-mask posterior is scalar-mode diagonal exactly in cases where
the mask term \(Q_M\) is diagonal.
\end{theorem}

\begin{proof}
Because \(K\) is positive definite in the theorem, the conditional Hamiltonian
is an ordinary finite quadratic form.  Expanding it, up to an additive term
independent of \(S\), gives
\begin{align}
    H(S|t)
    &=
    \frac{1}{2}S^TK^{-1}S
    +
    \frac{1}{2b^2}(t-HS)^T(t-HS)
    \nonumber\\
    &=
    \frac{1}{2}S^T\left(K^{-1}+b^{-2}H^TH\right)S
    -
    b^{-2}S^TH^Tt
    +
    \mathrm{const.}
\end{align}
Taking the Hessian of the conditional Hamiltonian with respect to \(S\) gives
\begin{equation}
    \nabla_S^2 H(S|t)
    =
    K^{-1}+b^{-2}H^TH
    =
    K^{-1}+b^{-2}D_M .
\end{equation}
Thus the Hessian is exactly the posterior precision on the full latent lattice.
All mask geometry in this precision enters through \(D_M=H^TH\).  We write it as
\begin{equation}
    P_M
    =
    K^{-1}+b^{-2}D_M .
\end{equation}
Completing the square in this quadratic Hamiltonian then identifies the
posterior mean and covariance.  Specifically,
\begin{equation}
    H(S|t)
    =
    \frac{1}{2}(S-\mu_M)^TP_M(S-\mu_M)
    +
    \mathrm{const.},
\end{equation}
where
\begin{equation}
    \mu_M
    =
    P_M^{-1}b^{-2}H^Tt .
\end{equation}
Therefore the square-completion calculation gives the exact Gaussian posterior
covariance
\begin{equation}
    \Sigma_M=P_M^{-1}
    =
    \left(K^{-1}+b^{-2}D_M\right)^{-1},
\end{equation}
and the posterior mean \(\mu_M=\Sigma_M b^{-2}H^Tt\).

The observation-space forms follow by applying the Woodbury identity, also
called the matrix inversion lemma, to \(K^{-1}+b^{-2}H^TH\)
\cite{Harville1997,GolubVanLoan2013}.  Appendix~\ref{app:woodbury-identity}
records this algebraic step in the form used for computation.  This produces an
algebraically identical expression whose matrix inverse is the \(m\times m\)
observation-space matrix \(HKH^T+b^2I_m\):
\begin{equation}
    \Sigma_M
    =
    K-KH^T(HKH^T+b^2I_m)^{-1}HK,
\end{equation}
and the corresponding posterior mean is
\begin{equation}
    \mu_M
    =
    KH^T(HKH^T+b^2I_m)^{-1}t.
\end{equation}
Finally, multiplying the latent-space precision by the Fourier transform gives
\(F(K^{-1}+b^{-2}D_M)F^\ast=\Lambda^{-1}+b^{-2}F D_MF^\ast\), which is the stated
coupled Fourier representation.
\end{proof}

\begin{remark}
The theorem is conditional on a fixed realized mask.  Independent scalar Fourier
modes occur in the diagonal-mask boundary cases.  Random masks introduced later
are outer generation models for realized masks; after a mask is realized, the
posterior precision is determined by its corresponding \(D_M\).
\end{remark}

\section{Fourier Representation and Mode Coupling}

In the complete-observation formulation, the prior precision and the
likelihood precision are diagonal in the same Fourier basis.  This simultaneous
diagonalization is what reduces the posterior calculation to independent scalar
Fourier modes.  In the sparse-observation formulation, the prior keeps this
property, but the likelihood precision changes from \(b^{-2}I_L\) to
\(b^{-2}D_M\).  This section makes the realized mask geometry visible in Fourier
space by writing the mask contribution as \(Q_M=F D_MF^\ast\).  Proposition~1 is
the explicit finite-mask statement that this matrix generally converts the
complete-observation scalar-mode posterior into a coupled finite-mode posterior.

Define
\begin{equation}
    \widetilde{\Sigma}_M=F\Sigma_MF^\ast .
\end{equation}
Using \(K=F^\ast\Lambda F\), the posterior precision matrix in Fourier space is
\begin{equation}
    \widetilde{\Sigma}_M^{-1}
    =
    \Lambda^{-1}+b^{-2}F D_MF^\ast .
\end{equation}
Define
\begin{equation}
    Q_M=F D_MF^\ast .
\end{equation}
Its matrix elements are
\begin{equation}
    (Q_M)_{nn'}
    =
    \frac{1}{L}
    \sum_{j=0}^{L-1}M_j e^{-i(\kappa_n-\kappa_{n'})j}.
\end{equation}
Equivalently, define the finite Fourier components of the realized mask by
\begin{equation}
    \widehat{M}_q
    =
    \sum_{j=0}^{L-1}M_j\exp\left(-\frac{2\pi i qj}{L}\right),
    \qquad q\in \mathbb{T}_L .
\end{equation}
Since \(\kappa_n-\kappa_{n'}=2\pi(n-n')/L\), the mask term satisfies
\begin{equation}
    (Q_M)_{nn'}=\frac{1}{L}\widehat{M}_{n-n'},
\end{equation}
where the index difference is understood modulo \(L\).  Thus the off-diagonal
elements of \(Q_M\) are not arbitrary matrix elements; they are precisely the
finite Fourier coefficients of the realized observation mask.  The observation
geometry therefore enters the posterior precision through the spectrum of the
mask.  A mask with a nonzero Fourier component at the difference frequency
\(n-n'\) couples the prior modes \(n\) and \(n'\).

The diagonal entries are
\begin{equation}
    (Q_M)_{nn}=\frac{m}{L}.
\end{equation}
This identity separates the observed-site count from the realized geometry of
the mask.  The fraction \(m/L\), or equivalently the missing rate \(p=1-m/L\),
fixes only the diagonal entries of \(Q_M\).  It does not determine the
off-diagonal entries
\begin{equation}
    (Q_M)_{nn'}=\frac{1}{L}\widehat{M}_{n-n'} \qquad (n\ne n'),
\end{equation}
which are the finite Fourier components of the realized mask at nonzero
difference frequencies.  These components encode how the observed sites are
spaced, clustered, or otherwise arranged on the finite torus.

The posterior trace is consequently not a function of the eigenvalue list of
\(Q_M\) alone.  By trace invariance,
\begin{equation}
    E(M)=\frac{1}{L}\Tr\Sigma_M
    =\frac{1}{L}\Tr\left[
    \left(\Lambda^{-1}+b^{-2}Q_M\right)^{-1}\right].
\end{equation}
Thus the finite posterior uncertainty depends on how the mask matrix \(Q_M\),
written in the Fourier basis, is positioned relative to the prior spectrum
\(\Lambda\).  The diagonal part records the missing rate, but the off-diagonal
Fourier components determine the mode couplings that enter the inverse trace.
In the complete-observation case, \(M_j=1\) for all \(j\), so
\(\widehat{M}_q=L\delta_{q0}\) and the off-diagonal mode coupling disappears.
For an arbitrary nonconstant mask, nonzero components with \(q\ne0\) generally
remain and couple Fourier modes.

\begin{proposition}[Fourier mode coupling by a nontrivial mask]
For a binary full-lattice mask, \(Q_M=F D_MF^\ast\) is diagonal in the Fourier
basis if and only if \(D_M=I_L\) or \(0\).  Hence any true sparse mask with
\(0<m<L\) has at least one nonzero off-diagonal Fourier component
\begin{equation}
    \widehat{M}_{n-n'}=
    \sum_{j=0}^{L-1}M_j e^{-i(\kappa_n-\kappa_{n'})j}
    \qquad(n\ne n'),
\end{equation}
and it couples Fourier modes in the posterior precision
\(\widetilde{\Sigma}_M^{-1}=\Lambda^{-1}+b^{-2}Q_M\).
\end{proposition}

\begin{proof}
If \(D_M=I_L\), then \(Q_M=I_L\); if \(D_M=0\), then \(Q_M=0\).  These are the
complete and empty diagonal cases.  Conversely, suppose that \(Q_M\) is diagonal.
Then \(D_M=F^\ast Q_MF\) is a circulant matrix.  Since \(D_M\) is also diagonal,
it must be a scalar multiple of \(I_L\).  The binary condition on the diagonal
entries then gives either \(D_M=I_L\) or \(0\).  Thus any binary mask with
\(0<m<L\) necessarily leaves off-diagonal entries in \(Q_M\), equivalently
nonzero finite Fourier components at some nonzero difference frequencies.
\end{proof}

The proposition identifies the finite-matrix mechanism behind the deterministic
comparisons below.  The missing rate fixes \((Q_M)_{nn}=m/L\), whereas the
realized mask geometry fixes the off-diagonal Fourier components.  These
components are precisely the entries that convert the complete-observation
scalar-mode posterior into the coupled inverse-trace quantity \(E(M)\).  Hence
two masks with the same \(p\) can have different posterior covariances and
different posterior traces, as in Table~\ref{tab:deterministic-mask-geometry}.

\begin{remark}[Mask spectrum and relative orientation]
Although \(Q_M\) is unitarily similar to \(D_M\), its eigenvalue list contains only
the observed-site count: \(m\) eigenvalues are one and \(L-m\) eigenvalues are
zero.  Therefore all masks with the same \(m\) have the same eigenvalues of
\(Q_M\).  The posterior trace, however, is
\begin{equation}
    \frac{1}{L}\Tr\left[
    \left(\Lambda^{-1}+b^{-2}Q_M\right)^{-1}\right],
\end{equation}
and this trace depends on the relative orientation of \(Q_M\) and the diagonal
prior spectrum \(\Lambda\), not on the eigenvalues of \(Q_M\) alone.  Equivalently,
\(\Lambda^{-1}\) selects the Fourier basis as a preferred basis, and the entries
of \(Q_M\) in that basis are the finite mask Fourier components
\(L^{-1}\widehat{M}_{n-n'}\).  For a smooth prior, the larger eigenvalues of
\(\Lambda\) often lie in low-frequency modes.  The amount of low-frequency mask
spectral weight and its coupling to other modes can therefore change the inverse
trace even when the missing rate is fixed.  This is the Fourier-space reason why
regular and clustered masks with the same \(m\) can leave different posterior
uncertainties.
\end{remark}

Figure~\ref{fig:fourier-mode-coupling} summarizes the structural distinction
between the complete-observation case and the sparse-mask case.

\subsection{Complete-observation reduction}

When \(D_M=I_L\), the posterior covariance becomes diagonal in Fourier space:
\begin{equation}
    \widetilde{\Sigma}_{\mathrm{full},n}
    =
    \left(\lambda_n^{-1}+b^{-2}\right)^{-1}
    =
    \frac{\lambda_n b^2}{\lambda_n+b^2}.
\end{equation}
The posterior mean has the scalar coefficient
\begin{equation}
    \alpha_n
    =
    \frac{\lambda_n}{\lambda_n+b^2}.
\end{equation}
Thus the complete-observation formula of the previous
formulation~\cite{Tsuzurugi2025} is retained exactly as the boundary case
\(D_M=I_L\).

\begin{corollary}[Complete-observation matched error]
For the complete-observation mask \(D_M=I_L\), the matched reconstruction error is
\begin{equation}
    E_{\mathrm{full}}
    =
    \frac{1}{L}\sum_{n=0}^{L-1}
    \frac{\lambda_n b^2}{\lambda_n+b^2}.
\end{equation}
\end{corollary}

\section{Matched Mean Squared Error}

In the previous complete-observation formulation, the matched error
\(E_{1\min}\) was obtained as a Fourier-mode-wise sum because the posterior
covariance was diagonal in the Fourier basis.  Corollary~1 is the corresponding
complete-observation reduction in the present notation.  For a sparse mask,
\(\widetilde{\Sigma}_M\) is generally non-diagonal in Fourier space, so the
matched error is represented by the full posterior trace rather than by an
independent scalar-mode sum.  The natural finite-dimensional continuation of the
previous \(E_{1\min}\) is therefore the normalized posterior trace defined below.  In the complete-observation limit
\(D_M=I_L\), this trace reduces to the mode-wise expression in Corollary~1 and
hence to the previous complete-observation expression.

For a fixed realized mask \(M\), define
\begin{equation}
    E(M)=\frac{1}{L}\Tr\Sigma_M.
\end{equation}
This quantity is the Bayes risk of the matched posterior-mean estimator after
averaging over both the latent Gaussian field and the observation noise under the
same model.  Since the posterior covariance in a linear Gaussian model is independent
of the realized data values \(t\), the matched error depends on the mask and the
model parameters through \(\Sigma_M\).

Using the Woodbury form,
\begin{equation}
    E(M)
    =
    \frac{1}{L}\Tr K
    -
    \frac{1}{L}\Tr\left[
        KH^T(HKH^T+b^2I_m)^{-1}HK
    \right].
\end{equation}
The second term is the variance reduction due to the observed sites.  This form
also shows why an \(m\times m\) inverse is useful when the number of observations
is much smaller than \(L\).

\begin{proposition}[Monotonicity under adding observations]
Let \(M\) and \(M'\) be two masks with \(D_{M'}-D_M\) positive semidefinite.
Then
\begin{equation}
    \Sigma_{M'}\preceq \Sigma_M,
    \qquad
    E(M')\le E(M).
\end{equation}
\end{proposition}

\begin{proof}
The corresponding posterior precisions satisfy
\begin{equation}
    K^{-1}+b^{-2}D_{M'}
    \succeq
    K^{-1}+b^{-2}D_M .
\end{equation}
For positive definite matrices, larger precision implies smaller covariance in
the Loewner order.  Taking the trace gives the result.  This is consistent with
the variance-reduction representation above: adding observed sites increases the
non-negative reduction from the prior trace and therefore reduces posterior
variance.
\end{proof}

\begin{corollary}[Bounds between complete and empty observations]
Let
\begin{equation}
    E_{\mathrm{prior}}=\frac{1}{L}\operatorname{Tr}K
\end{equation}
be the empty-observation error, and let \(E_{\mathrm{full}}\) be the
complete-observation error in Corollary~1.  Then, for any observation mask,
\begin{equation}
    E_{\mathrm{full}}
    \le
    E(M)
    \le
    E_{\mathrm{prior}} .
\end{equation}
\end{corollary}

\begin{proof}
The claim follows from Proposition~2 by comparing an arbitrary mask with the
complete-observation mask \(D_M=I_L\) and the empty-observation mask \(D_M=0\).
\end{proof}

The bounds show that the finite-mask error is continuously ordered between the
complete- and empty-observation limits.  A practical threshold can be
introduced by specifying an external tolerance, for example by declaring a
reconstruction acceptable when
\begin{equation}
    E(M)<\varepsilon.
\end{equation}
The tolerance \(\varepsilon\) is supplied by the application or by an additional
limiting problem.

\begin{remark}[Continuous degradation]
For fixed finite \(L\), fixed covariance \(K\), and finite noise variance
\(b^2>0\), sparse observation changes the posterior covariance through a
finite-dimensional matrix formula.  The intrinsic quantity is the continuous
trace indicator
\begin{equation}
    E(M)=\frac{1}{L}\operatorname{Tr}\Sigma_M.
\end{equation}
A binary decision can be added by specifying an external tolerance or another
criterion outside the posterior formula itself.
\end{remark}

\section{Random Masks and Random Observation Intervals}

The posterior formula is conditional on a realized mask.  Random masks are used
here as a generation and reporting layer: they specify a distribution over
finite masks, and the posterior trace is evaluated after each mask has been
realized.

As a standard example, consider independent Bernoulli observation,
\begin{equation}
    M_j\sim \operatorname{Bernoulli}(\rho),
    \qquad
    0\le \rho\le 1,
\end{equation}
independently over \(j\).  Here \(\rho\) is the Bernoulli observation rate.  For
each realized mask, the observed-site count is \(m=|\Omega|\) and the realized
missing rate is \(p(M)=1-m/L\); the quantity \(1-\rho\) is the nominal missing
rate of the generation model.  For each realized mask,
\begin{equation}
    E(M)
    =
    \frac{1}{L}\Tr\left[
        \left(K^{-1}+b^{-2}D_M\right)^{-1}
    \right]
\end{equation}
is the deterministic finite-mask posterior trace.  The random-mask model
induces
the outer average
\begin{equation}
    \overline{E}(\rho)=\E_M[E(M)],
\end{equation}
which averages these deterministic posterior traces over realized masks.  The
mask averaging is performed outside the matrix inverse, after each finite mask
has been evaluated.

\begin{proposition}[Mean-mask approximation and exact averaging]
Because the matrix inverse is nonlinear,
\begin{equation}
    \mathbb{E}_M
    \left[
    \left(K^{-1}+b^{-2}D_M\right)^{-1}
    \right]
    \ne
    \left(K^{-1}+b^{-2}\mathbb{E}_M[D_M]\right)^{-1}
\end{equation}
in general.  For Bernoulli observation,
\begin{equation}
    \mathbb{E}_M[D_M]=\rho I,
\end{equation}
so replacing \(D_M\) by \(\rho I\) gives a mean-mask approximation rather than
the exact mask-averaged posterior covariance.
\end{proposition}

For one-dimensional cyclic masks, observation intervals may also be summarized
by positive integer gaps \(G_1,\ldots,G_m\) between successive observed sites,
with
\begin{equation}
    \sum_{s=1}^{m}G_s=L.
\end{equation}
This gap description is a diagnostic of the realized geometry.  Once \(M\) is
fixed, the posterior covariance remains
\begin{equation}
    \Sigma_M=
    \left(
        K^{-1}+b^{-2}D_M
    \right)^{-1}.
\end{equation}
For \(m=0\), \(\Sigma_M=K\) and \(E(M)=E_{\mathrm{prior}}\).  For \(m\ge1\),
we report
\begin{equation}
    G_{\max}=\max_s G_s,
    \qquad
    p(M)=1-\frac{m}{L},
\end{equation}
together with the normalized trace
\begin{equation}
    \frac{E(M)}{E_{\mathrm{prior}}},
    \qquad
    E_{\mathrm{prior}}=\frac{1}{L}\Tr K.
\end{equation}

In the circumferential-array interpretation, the corresponding largest angular
interval between successive available microphones is
\begin{equation}
    \Delta\theta_{\max}=\frac{2\pi G_{\max}}{L}.
\end{equation}
The pair
\((p(M),G_{\max})\) is a compact reporting diagnostic that complements the
realized-mask posterior covariance.

\section{Synthetic Circumferential Microphone-array Benchmark}

This section recasts the finite mask-geometry calculations as a dimensionless
synthetic benchmark for missing-channel reconstruction in the circumferential
microphone array of Sect.~1.1.  Consider \(L=64\) nominal microphones equally
spaced around one circular-duct cross-section at
\begin{equation}
    \theta_j=\frac{2\pi j}{64},
    \qquad
    \Delta\theta=\frac{2\pi}{64}=5.625^\circ .
\end{equation}
The benchmark is not based on measured duct data, and the covariance parameters
below are not calibrated to a particular acoustic installation.  Its purpose is
to isolate the effect of channel-availability geometry while retaining the
exact finite posterior covariance derived above.

We report the realized-mask trace \(E(M)=L^{-1}\Tr\Sigma_M\), interpreted as the
residual prediction variance averaged over all 64 nominal microphone positions,
the realized unavailable-channel fraction \(p(M)=1-m/L\), and the largest cyclic
gap \(G_{\max}\) between successive usable channels.  In angular units,
\begin{equation}
    \Delta\theta_{\max}
    =\frac{2\pi G_{\max}}{L}
    =5.625^\circ G_{\max}.
\end{equation}
For deterministic comparisons, we also examine the finite Fourier spectrum of
the availability mask.  For Bernoulli summaries, \(\rho\) denotes the
channel-availability rate and \(1-\rho\) is used only as the nominal unavailable
fraction.

For each prescribed parameter setting, the finite procedure is: construct the
admissible covariance, prescribe or generate a channel-availability mask, form
\(H\), compute
\begin{equation}
    \Sigma_M=K-KH^T(HKH^T+b^2I_m)^{-1}HK,
\end{equation}
and report \(E(M)=L^{-1}\Tr\Sigma_M\) together with \(p(M)=1-m/L\) and, when
\(m\ge1\), \(G_{\max}\) or its angular equivalent
\(\Delta\theta_{\max}\).

The covariance \(K\) is the minimum-image Gaussian constructed from
Eq.~\eqref{eq:minimum-image-gaussian}
with \(L=64\), \(a=1\), and \(\ell=3\), used in the admissible range satisfying
Eq.~\eqref{eq:minimum-image-admissibility}.  Here \(a=1\) normalizes the pressure-fluctuation variance to
\(K_{jj}=a^2=1\), while the lattice correlation length \(\ell=3\) corresponds to
the angular correlation length
\begin{equation}
    \ell_\theta=\frac{2\pi\ell}{L}
    \simeq 0.295\ \mathrm{rad}
    \simeq 16.9^\circ .
\end{equation}
The value \(b^2=1\) is interpreted as the normalized measurement-noise
variance.  These dimensionless choices preserve the previous covariance and
all reported numerical results; they are not claimed to constitute an
acoustically calibrated covariance model.

For every realized mask, the posterior covariance is evaluated by the Woodbury
form above.  Figure~\ref{fig:regular-sparse-reconstruction} shows one
posterior-mean reconstruction with 32 evenly distributed usable channels.
Table~\ref{tab:numerical-results} and
Fig.~\ref{fig:posterior-trace-summary} are obtained by generating 100 Bernoulli
channel-availability masks for each \(\rho\) and then computing \(E(M)\)
separately for each realized mask.  Table~\ref{tab:deterministic-mask-geometry}
compares two deterministic masks: the regular alternating mask, in which every
other channel is usable, and the clustered missing block, in which 32 adjacent
channels are unavailable and the other 32 adjacent channels are usable.  The
latter describes contiguous groupwise missingness without assigning it to a
particular failure mechanism.  Figure~\ref{fig:gap-posterior-scatter} uses a
different generation layer from Table~\ref{tab:numerical-results}: it fixes
\(m=32\) and plots fixed-count random availability masks together with the same
two deterministic reference masks.  Thus the Bernoulli average in
Table~\ref{tab:numerical-results} and the fixed-count scatter in
Fig.~\ref{fig:gap-posterior-scatter} are not the same random experiment; both
are reported only to make the finite realized-mask dependence of \(E(M)\)
transparent.  This paragraph fixes only the numerical protocol and adds no new
theory.

The single-realization display in
Fig.~\ref{fig:regular-sparse-reconstruction} identifies the latent pressure
fluctuation, observed noisy microphone values, and posterior mean around the
circumference for one regular alternating mask. Its primary horizontal axis is
the lattice-site coordinate \(j=0,\ldots,63\), and its lower secondary axis uses
\(\theta_j=360^\circ j/64\); hence the sites \(j=0,16,32,48\) correspond to
\(0,90,180,270^\circ\), respectively. The posterior-trace
dependence on channel geometry is summarized in
Tables~\ref{tab:numerical-results} and
\ref{tab:deterministic-mask-geometry} and in
Figs.~\ref{fig:posterior-trace-summary} and
\ref{fig:gap-posterior-scatter}.

Table~\ref{tab:numerical-results} records finite posterior-trace computations
for Bernoulli channel-availability masks.  As the availability rate \(\rho\)
decreases, the nominal unavailable fraction \(1-\rho\) increases and the average
posterior trace increases.  Thus reconstruction uncertainty grows when fewer
microphone channels are usable, consistently with the monotonicity statement in
Proposition~2.  Each unit of \(G_{\max}\) or its standard deviation corresponds
to an azimuthal interval of \(5.625^\circ\).

\begin{table}[tb]

\centering
\caption{
Results from the synthetic circumferential microphone-array benchmark for
Bernoulli channel-availability masks, with \(L=64\), \(a=1\), \(\ell=3\), and
\(b^2=1\).  The averages and standard deviations are over 100 mask
realizations, except that \(\rho=1\) is deterministic and its standard
deviations are therefore zero.  The column \(1-\rho\) is the nominal
unavailable-channel fraction of the Bernoulli generation model; each
realization is evaluated with its own usable-channel count \(m=|\Omega|\) and
realized unavailable fraction \(p(M)=1-m/L\).  Values of \(G_{\max}\) are in
channel intervals; one interval is \(5.625^\circ\).  These are synthetic,
dimensionless results rather than measurements from a duct experiment.
}
\label{tab:numerical-results}
\begin{tabular}{c c c c c c c}
\hline
\(L\) & \(\rho\) & \(1-\rho\) &
\(\langle G_{\max}\rangle\) & \(\operatorname{sd}(G_{\max})\) &
\(\langle E(M)\rangle\) & \(\operatorname{sd}(E(M))\) \\
\hline
64 & 1.0 & 0.0 & 1.00 & 0.00 & 0.1957 & 0.0000 \\
64 & 0.8 & 0.2 & 3.38 & 0.80 & 0.2381 & 0.0133 \\
64 & 0.5 & 0.5 & 6.35 & 1.95 & 0.3488 & 0.0390 \\
64 & 0.2 & 0.8 & 15.60 & 5.77 & 0.6134 & 0.0774 \\
\hline
\end{tabular}

\end{table}

Table~\ref{tab:deterministic-mask-geometry} gives the complementary fixed-rate
comparison.  With \(p(M)=0.5\), the regular alternating mask and the contiguous
32-channel missing block have the same number of usable channels but different
posterior traces: \(0.3116\) and \(0.5643\), respectively.  The latter is about
\(1.8\) times the regular-mask value.  Thus, even with the same 32 unavailable
channels, their circumferential arrangement changes the full-array average
reconstruction uncertainty by a practically substantial factor in this
synthetic benchmark.  Their largest usable-channel gaps are
\(\Delta\theta_{\max}=11.25\) and \(185.625^\circ\), respectively.

\begin{table}[tb]

\centering
\caption{
Deterministic comparison in the synthetic circumferential microphone-array
benchmark with \(L=64\), \(a=1\), \(\ell=3\), and \(b^2=1\).  Both masks have
the same unavailable-channel fraction \(p(M)=1-m/L=0.5\).  The regular mask
uses every other microphone, whereas the clustered mask leaves 32 adjacent
channels unavailable without specifying a failure cause.  The bracketed angle
is \(\Delta\theta_{\max}=2\pi G_{\max}/L\), the largest azimuthal interval
between usable channels.  The residual prediction variance \(E(M)\) is averaged
over all nominal microphone positions.  The values are synthetic and
dimensionless, not measured acoustic data.
}
\label{tab:deterministic-mask-geometry}
\begin{tabular}{c c l c c}
\hline
\(L\) & \(p(M)\) & Channel geometry & \(G_{\max}\) & \(E(M)\) \\
\hline
64 & 0.5 & alternating 32 usable & 2 & 0.3116 \\
64 & 0.5 & contiguous 32 missing & 33 & 0.5643 \\
\hline
\end{tabular}

\end{table}

Figure~\ref{fig:posterior-trace-summary} summarizes the Bernoulli results in
Table~\ref{tab:numerical-results}: the mean residual prediction variance
increases with the nominal unavailable-channel fraction \(1-\rho\).
Figure~\ref{fig:gap-posterior-scatter} keeps \(p(M)=0.5\) fixed and plots
\(E(M)/E_{\mathrm{prior}}\) against \(G_{\max}\), with
the two deterministic masks from Table~\ref{tab:deterministic-mask-geometry} as
references.  The random masks in this scatter are fixed-count random masks with
exactly \(m=32\) usable channels, generated under the constraint \(m/L=0.5\),
rather than Bernoulli masks.  It should be read as a finite diagnostic display:
at the same usable-channel count, the azimuthal gap geometry supplies
information about the posterior trace beyond the unavailable fraction.  In
this finite fixed-count case study, larger angular gaps tend to accompany
larger normalized posterior traces.

The deterministic comparison can also be read directly in Fourier space.  Define
the finite mask components
\begin{equation}
    \widehat{M}_q
    =
    \sum_{j=0}^{L-1} M_j\exp\left(-\frac{2\pi i qj}{L}\right),
    \qquad q\in\mathbb{T}_L .
\end{equation}
Then
\begin{equation}
    (Q_M)_{nn'}=\frac{1}{L}\widehat{M}_{n-n'}
\end{equation}
with the index difference understood modulo \(L\).  The two masks in
Table~\ref{tab:deterministic-mask-geometry} have the same value of \(m\), so
\(Q_M\) has the same eigenvalue list in both cases.  What differs is the mask
spectrum in the Fourier basis in which \(\Lambda\) is diagonal.  For the regular
alternating mask at \(L=64\), the nonzero mask components are concentrated at
\(q=0\) and the Nyquist frequency \(q=32\).  For the clustered block, the mask
spectrum is broad and contains large low-frequency components.  Since the smooth
Gaussian prior has large low-frequency eigenvalues, this difference changes the
inverse trace of \(\Lambda^{-1}+b^{-2}Q_M\) and gives a direct Fourier-space
interpretation of the larger posterior trace in
Table~\ref{tab:deterministic-mask-geometry}.

\section{Discussion}

For the circumferential microphone-array example, the main structural
consequence of the mask-conditioned Fourier formulation is the separation
between the number of usable channels and the Fourier geometry carried by
\(Q_M\).  Prior diagonalization survives in periodic Fourier coordinates,
whereas posterior mode independence generally does not.  The count fixes only
the diagonal entries \((Q_M)_{nn}=m/L\), while the off-diagonal entries
\((Q_M)_{nn'}=L^{-1}\widehat{M}_{n-n'}\) couple prior modes in
\(\Lambda^{-1}+b^{-2}Q_M\).  Missing-channel reconstruction on the idealized
ring is therefore a finite mode-coupling problem rather than a collection of
independent scalar shrinkage problems.

For fixed model parameters and a fixed mask, the posterior covariance
\(\Sigma_M\) is independent of the observed values \(t\).  Hence
\(E(M)=L^{-1}\Tr\Sigma_M\) can be computed once the mask is known and used as a
mask-conditioned uncertainty indicator for comparing missing-channel patterns
within this sensor-ring model.

The synthetic 64-channel benchmark shows this dependence without relying on an
asymptotic argument.  At the same unavailable-channel fraction \(p(M)=0.5\),
Table~\ref{tab:deterministic-mask-geometry} gives posterior traces \(0.3116\)
and \(0.5643\) for the regular alternating mask and the contiguous missing
block.  The latter uncertainty is about \(1.8\) times the former even though
both configurations contain 32 unavailable channels.  Therefore the number of
unavailable microphones alone does not determine reconstruction uncertainty;
their circumferential placement must also be retained.  These values are a
dimensionless synthetic comparison and are not a performance claim for an
operating duct.

The Fourier representation explains the difference.  Although the eigenvalues of
\(Q_M=F D_MF^\ast\) are fixed by the count \(m\), the posterior trace is the
inverse trace of \(\Lambda^{-1}+b^{-2}Q_M\).  It therefore depends on the
orientation of the mask matrix in the Fourier basis selected by the prior
spectrum \(\Lambda\).  This is the finite-matrix mechanism by which mask geometry
changes spectral reconstruction uncertainty.

This viewpoint also clarifies the relation to neighboring methods.  In
Wiener-filter terminology, the complete-observation periodic case is the
independent scalar Fourier-shrinkage limit.  In the sparse-mask problem, the
central object is the non-diagonal mask matrix \(Q_M\) inside the posterior
precision.  Random masks and gap-sensitive quantities are reporting layers
outside the conditional posterior formula: a generation model defines an outer
average over realized masks, and diagnostics such as \((p(M),G_{\max})\) record
geometry beyond the missing rate.

The duct interpretation depends on equally spaced sensors on a circular
cross-section, a rotationally stationary pressure covariance, and homogeneous
independent sensor noise.  For a noncircular duct, an unequally spaced sensor
layout, a circumferentially nonuniform mean flow, or azimuthally varying wall
conditions, the covariance need not be circulant.  Sensor-dependent or
correlated noise likewise replaces the common-noise likelihood used here.
Under those conditions the scalar Fourier-mode structure and the precision
\(\Lambda^{-1}+b^{-2}Q_M\) do not remain valid without modification, although
the general finite-dimensional Gaussian conditioning identity still applies.
Accordingly, the synthetic results should not be generalized to an operating
duct before these assumptions are checked.

\section{Conclusion}

Motivated by missing channels in a circumferential microphone array at one
circular-duct cross-section, we have formulated the finite linear-Gaussian
posterior for an arbitrary realized mask on an equally spaced periodic sensor
ring.  The standard conditioning identity gives the posterior mean and
covariance for each fixed mask; the contribution here is the explicit finite
mask-geometry representation, not a new Gaussian-conditioning theorem or a
duct-mode identification method.

The complete-observation scalar Fourier formula is recovered at \(D_M=I_L\).
Sparse masks introduce \(Q_M=F D_MF^\ast\), whose off-diagonal finite Fourier
components couple modes in \(\Lambda^{-1}+b^{-2}Q_M\).  The matched error becomes
\(E(M)=L^{-1}\Tr\Sigma_M\), a mask-conditioned posterior-trace indicator.  The
figures display this mechanism through the structural Fourier schematic, the
Bernoulli-mask trace summary, and the fixed-rate gap diagnostic.

In the synthetic 64-channel circumferential-array benchmark, the same 50\%
unavailable-channel fraction gives \(E(M)=0.3116\) for alternating missing
channels and \(E(M)=0.5643\) for one contiguous 32-channel missing block.  Thus,
within the stated ring model, the unavailable fraction is a coarse descriptor
and the realized circumferential mask geometry must be retained.

These conclusions are limited to a circular, equally instrumented ring with a
circulant prior covariance and common independent sensor noise.  Noncircular
geometry, unequal sensor spacing, circumferentially nonuniform mean flow, or
sensor-dependent noise generally destroys that circulant covariance or the
common-noise likelihood, so the scalar Fourier-mode structure used here does
not carry over unchanged.  The analysis is therefore an exact finite reference
calculation under explicit idealizations, not a validation of missing-channel
reconstruction performance for fan or compressor ducts in general.

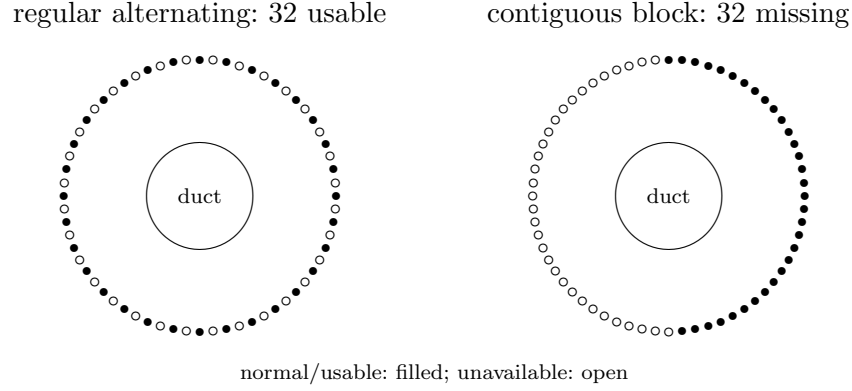
\begin{figure}[p]

\centering
\setlength{\unitlength}{1mm}
\begin{picture}(130,55)
\put(34,25){\circle{32}}
\put(96,25){\circle{32}}
\put(34,25){\makebox(0,0){\scriptsize duct}}
\put(96,25){\makebox(0,0){\scriptsize duct}}
\put(34,49){\makebox(0,0){regular alternating: 32 usable}}
\put(96,49){\makebox(0,0){contiguous block: 32 missing}}
\put(34.00,43.00){\circle*{1.15}}
\put(35.76,42.91){\circle{1.15}}
\put(37.51,42.65){\circle*{1.15}}
\put(39.23,42.22){\circle{1.15}}
\put(40.89,41.63){\circle*{1.15}}
\put(42.49,40.87){\circle{1.15}}
\put(44.00,39.97){\circle*{1.15}}
\put(45.42,38.91){\circle{1.15}}
\put(46.73,37.73){\circle*{1.15}}
\put(47.91,36.42){\circle{1.15}}
\put(48.97,35.00){\circle*{1.15}}
\put(49.87,33.49){\circle{1.15}}
\put(50.63,31.89){\circle*{1.15}}
\put(51.22,30.23){\circle{1.15}}
\put(51.65,28.51){\circle*{1.15}}
\put(51.91,26.76){\circle{1.15}}
\put(52.00,25.00){\circle*{1.15}}
\put(51.91,23.24){\circle{1.15}}
\put(51.65,21.49){\circle*{1.15}}
\put(51.22,19.77){\circle{1.15}}
\put(50.63,18.11){\circle*{1.15}}
\put(49.87,16.51){\circle{1.15}}
\put(48.97,15.00){\circle*{1.15}}
\put(47.91,13.58){\circle{1.15}}
\put(46.73,12.27){\circle*{1.15}}
\put(45.42,11.09){\circle{1.15}}
\put(44.00,10.03){\circle*{1.15}}
\put(42.49,9.13){\circle{1.15}}
\put(40.89,8.37){\circle*{1.15}}
\put(39.23,7.78){\circle{1.15}}
\put(37.51,7.35){\circle*{1.15}}
\put(35.76,7.09){\circle{1.15}}
\put(34.00,7.00){\circle*{1.15}}
\put(32.24,7.09){\circle{1.15}}
\put(30.49,7.35){\circle*{1.15}}
\put(28.77,7.78){\circle{1.15}}
\put(27.11,8.37){\circle*{1.15}}
\put(25.51,9.13){\circle{1.15}}
\put(24.00,10.03){\circle*{1.15}}
\put(22.58,11.09){\circle{1.15}}
\put(21.27,12.27){\circle*{1.15}}
\put(20.09,13.58){\circle{1.15}}
\put(19.03,15.00){\circle*{1.15}}
\put(18.13,16.51){\circle{1.15}}
\put(17.37,18.11){\circle*{1.15}}
\put(16.78,19.77){\circle{1.15}}
\put(16.35,21.49){\circle*{1.15}}
\put(16.09,23.24){\circle{1.15}}
\put(16.00,25.00){\circle*{1.15}}
\put(16.09,26.76){\circle{1.15}}
\put(16.35,28.51){\circle*{1.15}}
\put(16.78,30.23){\circle{1.15}}
\put(17.37,31.89){\circle*{1.15}}
\put(18.13,33.49){\circle{1.15}}
\put(19.03,35.00){\circle*{1.15}}
\put(20.09,36.42){\circle{1.15}}
\put(21.27,37.73){\circle*{1.15}}
\put(22.58,38.91){\circle{1.15}}
\put(24.00,39.97){\circle*{1.15}}
\put(25.51,40.87){\circle{1.15}}
\put(27.11,41.63){\circle*{1.15}}
\put(28.77,42.22){\circle{1.15}}
\put(30.49,42.65){\circle*{1.15}}
\put(32.24,42.91){\circle{1.15}}
\put(96.00,43.00){\circle*{1.15}}
\put(97.76,42.91){\circle*{1.15}}
\put(99.51,42.65){\circle*{1.15}}
\put(101.23,42.22){\circle*{1.15}}
\put(102.89,41.63){\circle*{1.15}}
\put(104.49,40.87){\circle*{1.15}}
\put(106.00,39.97){\circle*{1.15}}
\put(107.42,38.91){\circle*{1.15}}
\put(108.73,37.73){\circle*{1.15}}
\put(109.91,36.42){\circle*{1.15}}
\put(110.97,35.00){\circle*{1.15}}
\put(111.87,33.49){\circle*{1.15}}
\put(112.63,31.89){\circle*{1.15}}
\put(113.22,30.23){\circle*{1.15}}
\put(113.65,28.51){\circle*{1.15}}
\put(113.91,26.76){\circle*{1.15}}
\put(114.00,25.00){\circle*{1.15}}
\put(113.91,23.24){\circle*{1.15}}
\put(113.65,21.49){\circle*{1.15}}
\put(113.22,19.77){\circle*{1.15}}
\put(112.63,18.11){\circle*{1.15}}
\put(111.87,16.51){\circle*{1.15}}
\put(110.97,15.00){\circle*{1.15}}
\put(109.91,13.58){\circle*{1.15}}
\put(108.73,12.27){\circle*{1.15}}
\put(107.42,11.09){\circle*{1.15}}
\put(106.00,10.03){\circle*{1.15}}
\put(104.49,9.13){\circle*{1.15}}
\put(102.89,8.37){\circle*{1.15}}
\put(101.23,7.78){\circle*{1.15}}
\put(99.51,7.35){\circle*{1.15}}
\put(97.76,7.09){\circle*{1.15}}
\put(96.00,7.00){\circle{1.15}}
\put(94.24,7.09){\circle{1.15}}
\put(92.49,7.35){\circle{1.15}}
\put(90.77,7.78){\circle{1.15}}
\put(89.11,8.37){\circle{1.15}}
\put(87.51,9.13){\circle{1.15}}
\put(86.00,10.03){\circle{1.15}}
\put(84.58,11.09){\circle{1.15}}
\put(83.27,12.27){\circle{1.15}}
\put(82.09,13.58){\circle{1.15}}
\put(81.03,15.00){\circle{1.15}}
\put(80.13,16.51){\circle{1.15}}
\put(79.37,18.11){\circle{1.15}}
\put(78.78,19.77){\circle{1.15}}
\put(78.35,21.49){\circle{1.15}}
\put(78.09,23.24){\circle{1.15}}
\put(78.00,25.00){\circle{1.15}}
\put(78.09,26.76){\circle{1.15}}
\put(78.35,28.51){\circle{1.15}}
\put(78.78,30.23){\circle{1.15}}
\put(79.37,31.89){\circle{1.15}}
\put(80.13,33.49){\circle{1.15}}
\put(81.03,35.00){\circle{1.15}}
\put(82.09,36.42){\circle{1.15}}
\put(83.27,37.73){\circle{1.15}}
\put(84.58,38.91){\circle{1.15}}
\put(86.00,39.97){\circle{1.15}}
\put(87.51,40.87){\circle{1.15}}
\put(89.11,41.63){\circle{1.15}}
\put(90.77,42.22){\circle{1.15}}
\put(92.49,42.65){\circle{1.15}}
\put(94.24,42.91){\circle{1.15}}
\put(65,1.5){\makebox(0,0){\scriptsize normal/usable: filled; unavailable: open}}
\end{picture}
\caption{
Mask-geometry schematic for the synthetic circumferential microphone-array
benchmark.  Each circular-duct cross-section contains 64 equally spaced
microphones at angular intervals of \(5.625^\circ\).  The left panel makes
every other channel usable; the right panel leaves one contiguous group of 32
channels unavailable.  Filled markers denote normal, usable channels and open
markers denote unavailable channels.  The cyclic ordering joins the top
channel to itself after one full turn.  The right panel represents groupwise
missingness without attributing it to a particular failure mechanism.
}
\label{fig:mask-geometry}

\end{figure}

\begin{figure}[p]

\centering
\begin{tabular}{c@{\qquad}c}
\fbox{%
\begin{minipage}{0.39\linewidth}
\centering
full observation\\[1mm]
\(D_M=I_L,\; Q_M=I_L\)\\[1mm]
\[
\widetilde{\Sigma}_{\mathrm{full}}^{-1}
=\Lambda^{-1}+b^{-2}I_L
\]
\[
\begin{pmatrix}
\ast&0&0&0\\
0&\ast&0&0\\
0&0&\ast&0\\
0&0&0&\ast
\end{pmatrix}
\]
\end{minipage}}
&
\fbox{%
\begin{minipage}{0.39\linewidth}
\centering
sparse mask\\[1mm]
\(Q_M=F D_M F^{\ast}\)\\[1mm]
\[
\widetilde{\Sigma}_M^{-1}
=\Lambda^{-1}+b^{-2}Q_M
\]
\[
\begin{pmatrix}
\ast&\ast&0&\ast\\
\ast&\ast&\ast&0\\
0&\ast&\ast&\ast\\
\ast&0&\ast&\ast
\end{pmatrix}
\]
\end{minipage}}
\\[2mm]
diagonal posterior precision & coupled posterior precision
\end{tabular}
\caption{
Fourier-space mode-coupling schematic.  Under complete observation,
\(D_M=I_L\) gives \(Q_M=I_L\), so the Fourier-space posterior precision is
diagonal.  Under a sparse mask, \(Q_M=F D_M F^{\ast}\) is generally
non-diagonal and couples Fourier modes.  The schematic summarizes the structural
distinction between complete and sparse observation.
}
\label{fig:fourier-mode-coupling}

\end{figure}

\begin{figure}[p]

\centering
\includegraphics[width=0.82\linewidth]{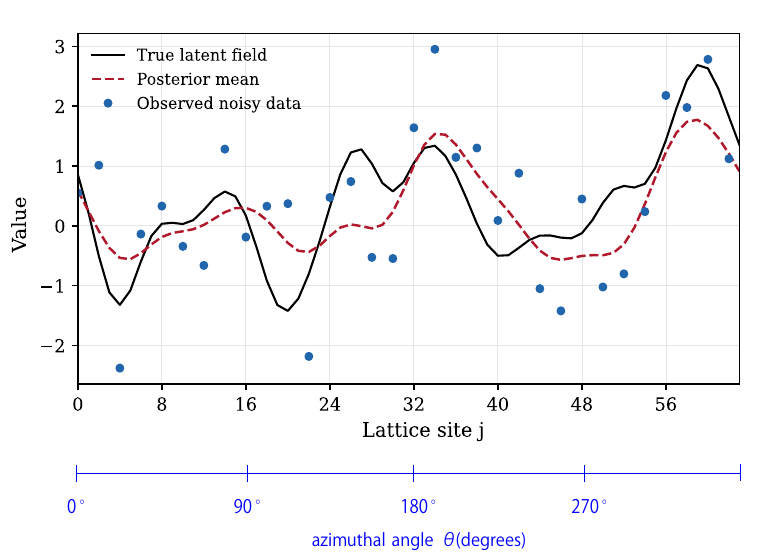}
\caption{
(Color online) Single-realization reconstruction for the synthetic circumferential
microphone-array benchmark with \(L=64\). The solid and dashed curves show
the latent pressure fluctuation \(S(\theta_j)\) and posterior mean
\(\mu_{M,j}\), respectively, and the markers show the observed noisy microphone
values \(t_j\) at the 32 usable sites of a regular alternating mask with
\(p(M)=0.5\). The primary horizontal axis gives the lattice-site index
\(j=0,\ldots,63\), and the lower axis gives the azimuthal angle
\(\theta_j=360^\circ j/64\); thus \(0\), \(90\), \(180\),
and \(270^\circ\) align with \(j=0\), \(16\), \(32\), and \(48\),
respectively. The data are synthetic and dimensionless.
}
\label{fig:regular-sparse-reconstruction}

\end{figure}

\begin{figure}[p]

\centering
\includegraphics[width=0.72\linewidth]{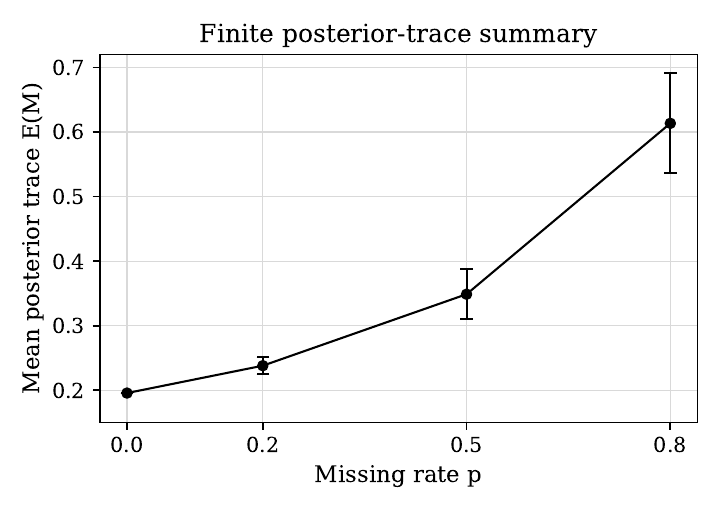}
\caption{
Posterior-trace summary for the synthetic circumferential microphone-array
benchmark based on Table~\ref{tab:numerical-results}.  The points show the mean
full-array residual prediction variance \(\langle E(M)\rangle\), and the vertical
bars show \(\operatorname{sd}(E(M))\), for \(L=64\), \(a=1\), \(\ell=3\), and
\(b^2=1\).  The horizontal coordinate is the nominal Bernoulli
unavailable-channel fraction \(1-\rho\); each posterior trace is computed after
a finite availability mask is realized.  The results are synthetic and
dimensionless, not measured acoustic data.
}
\label{fig:posterior-trace-summary}

\end{figure}

\begin{figure}[p]

\centering
\includegraphics[width=0.72\linewidth]{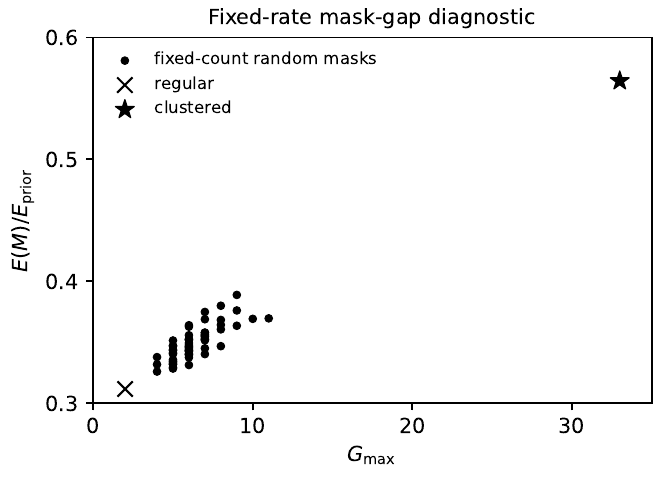}
\caption{
Gap-diagnostic scatter for the synthetic circumferential microphone-array
benchmark with \(L=64\), \(m=32\) (\(p(M)=1-m/L=0.5\)), \(a=1\), \(\ell=3\),
and \(b^2=1\).  Small filled circles are fixed-count random-mask realizations
with exactly \(m=32\) usable channels.  The regular alternating mask and the
contiguous 32-channel missing block from
Table~\ref{tab:deterministic-mask-geometry} are included as cross and star
reference symbols.  One unit on the \(G_{\max}\) axis corresponds to
\(5.625^\circ\), so the horizontal coordinate also represents the largest
azimuthal interval between usable microphones.  This synthetic, dimensionless
display indicates that larger angular gaps tend to accompany larger normalized
posterior traces in the finite fixed-count case study.
}
\label{fig:gap-posterior-scatter}

\end{figure}

\clearpage
\appendix

\section{Zero-eigenvalue Convention}
\label{app:zero-eigenvalue}

This appendix fixes the support convention for positive-semidefinite covariance
matrices.  The main theorem is stated for the positive-definite case, and the
convention below extends the notation to finite covariance supports.

The main text writes inverse-covariance expressions for the strictly positive
case.  If an admissible finite periodic covariance has zero eigenvalues, then
\(K\) is positive semidefinite and the Gaussian prior is supported on the range
of \(K\).  In that case, inverse expressions are interpreted on this support.  Equivalently, one may replace \(K^{-1}\) by the Moore--Penrose
pseudoinverse \(K^+\) and restrict the posterior calculation to the covariance
support.

Let
\begin{equation}
    K=F^\ast \Lambda F,
    \qquad
    \Lambda=\diag(\lambda_0,\ldots,\lambda_{L-1}),
\end{equation}
with \(\lambda_n\ge0\).  The support is spanned by the Fourier modes with
\(\lambda_n>0\).  On this support, \(\Lambda^{-1}\) is understood as
\begin{equation}
    \lambda_n^{-1}
    \quad\text{for}\quad \lambda_n>0.
\end{equation}
Zero-variance modes are deterministic under the prior and have zero posterior
variance.  The posterior covariance formula is applied on the positive-variance
subspace, and the covariance in zero-variance directions remains zero.

Thus the positive-semidefinite case is the same finite Gaussian posterior
calculation restricted to the covariance support.

\section{Woodbury Identity Form}
\label{app:woodbury-identity}

This appendix records the observation-space implementation form with an
\(m\times m\) inverse.  The identity is the standard Woodbury matrix identity,
or matrix inversion lemma, in the form used for finite linear-Gaussian
conditioning \cite{Harville1997,GolubVanLoan2013}.

For positive-definite \(K\) and \(b^2>0\),
\begin{equation}
\left(K^{-1}+b^{-2}H^TH\right)^{-1}
=K-KH^T(HKH^T+b^2I_m)^{-1}HK.
\end{equation}
Multiplying the right-hand side by \(K^{-1}+b^{-2}H^TH\) verifies the identity.
The form is useful because the inverse inside the observation-space expression
has dimension \(m\times m\), whereas the latent-space precision has dimension
\(L\times L\).

The corresponding posterior mean is
\begin{equation}
    \mu_M=KH^T(HKH^T+b^2I_m)^{-1}t.
\end{equation}
Thus the Woodbury form is algebraically the same posterior written in
observation-space coordinates.

\section{Coordinate-free Finite Operator Form}
\label{app:dimension-independent-form}

This appendix records the coordinate-free finite-dimensional algebra behind the
posterior identity used in the main text.  It is not a separate dimensional
extension of the one-dimensional periodic-lattice formulation.

Let a finite latent Gaussian vector \(S\in\mathbb{R}^{N}\) have covariance
\(K\), and let observations be selected by an arbitrary finite operator
\(H\in\mathbb{R}^{m\times N}\):
\begin{equation}
    t=HS+\xi,
    \qquad
    \xi\sim\mathcal{N}(0,b^2I_m).
\end{equation}
Then the same finite linear-Gaussian posterior identity gives
\begin{equation}
    \Sigma
    =
    \left(K^{-1}+b^{-2}H^TH\right)^{-1}
    =
    K-KH^T(HKH^T+b^2I_m)^{-1}HK.
\end{equation}
If a finite unitary coordinate transform \(U\) is chosen, for example the
finite unitary Fourier matrix corresponding to a finite periodic model, then
\begin{equation}
    U\Sigma^{-1}U^\ast
    =
    UK^{-1}U^\ast+b^{-2}UH^THU^\ast .
\end{equation}
When \(UKU^\ast\) is diagonal, the first term is diagonal and the second term is
the transformed mask precision.  A sparse observation operator generally makes
\(UH^THU^\ast\) non-diagonal, so the same finite algebra produces
coordinate-mode coupling in the transformed representation.

This appendix therefore records only a form-preserving finite operator identity
in finite dimensions; it does not assert a continuum limit, a large-system
limit, or a sampling-geometry theorem.


\begin{thebibliography}{99}

\bibitem{RasmussenWilliams}
C. E. Rasmussen and C. K. I. Williams, \textit{Gaussian Processes for Machine Learning}
(MIT Press, Cambridge, MA, 2006).

\bibitem{Bishop}
C. M. Bishop, \textit{Pattern Recognition and Machine Learning}
(Springer, New York, 2006).

\bibitem{MacKay}
D. J. C. MacKay, \textit{Information Theory, Inference, and Learning Algorithms}
(Cambridge University Press, Cambridge, U.K., 2003).


\bibitem{Joppa1987}
P. D. Joppa, \href{https://doi.org/10.2514/3.45482}{J. Aircr. \textbf{24}, 587 (1987)}.

\bibitem{WangEtAl2014}
L.-F. Wang, W.-Y. Qiao, L. Ji, and S.-Y. Yu,
\href{https://doi.org/10.13224/j.cnki.jasp.2014.04.024}{J. Aerosp. Power
\textbf{29}, 917 (2014)}.

\bibitem{WangEtAl2023}
R. Wang, M. Yu, Y. Bai, L. Yu, and G. Dong,
\href{https://doi.org/10.2514/1.J062596}{AIAA J. \textbf{61},
4018 (2023)}.


\bibitem{Cressie1993}
N. Cressie, \textit{Statistics for Spatial Data}
(Wiley, New York, 1993) revised ed.

\bibitem{Stein1999}
M. L. Stein, \textit{Interpolation of Spatial Data: Some Theory for Kriging}
(Springer, New York, 1999).

\bibitem{Wiener1949}
N. Wiener, \textit{Extrapolation, Interpolation, and Smoothing of Stationary
Time Series}
(MIT Press, Cambridge, MA, 1949).

\bibitem{OppenheimSchafer2010}
A. V. Oppenheim and R. W. Schafer, \textit{Discrete-Time Signal Processing}
(Prentice-Hall, Upper Saddle River, NJ, 2010) 3rd ed.

\bibitem{CandesWakin2008}
E. J. Cand\`es and M. B. Wakin,
\href{https://doi.org/10.1109/MSP.2007.914731}{IEEE Signal Process. Mag. \textbf{25},
21 (2008)}.

\bibitem{Tsuzurugi2025}
J. Tsuzurugi, \href{https://doi.org/10.7566/JPSJ.94.104801}{J. Phys. Soc. Jpn. \textbf{94}, 104801 (2025)}.

\bibitem{Harville1997}
D. A. Harville, \textit{Matrix Algebra from a Statistician's Perspective}
(Springer, New York, 1997).

\bibitem{GolubVanLoan2013}
G. H. Golub and C. F. Van Loan, \textit{Matrix Computations}
(Johns Hopkins University Press, Baltimore, MD, 2013) 4th ed.

\end{thebibliography}
\end{document}